\documentclass[runningheads]{llncs}

\usepackage{amsmath}
\usepackage{amssymb}
\usepackage{amsthm}
\newtheorem{openproblem}{Problem}
\theoremstyle{definition}
\newtheorem{observation}{Observation}

\usepackage[inline]{enumitem}

\usepackage{booktabs}

\usepackage{todonotes}

\begin{document}

\title{Guarding Quadrangulations and Stacked Triangulations with Edges\thanks{(Extended) abstracts of this paper were presented at the 36th European Workshop on Computational Geometry (EuroCG 2020)~\cite{Jungeblut2020EuroCG} and the 46th International Workshop on Graph-Theoretic Concepts in Computer Science (WG 2020)~\cite{Jungeblut2020WG}.}}
\author{Paul Jungeblut \and Torsten Ueckerdt}
\authorrunning{P. Jungeblut, T. Ueckerdt}
\institute{Karlsruhe Institute of Technology\\
\email{\{paul.jungeblut,torsten.ueckerdt\}@kit.edu}}
\maketitle

\begin{abstract}
  Let~$G = (V,E)$ be a plane graph.
  A face~$f$ of~$G$ is \emph{guarded} by an edge~$vw \in E$ if at least one vertex from~$\{v,w\}$ is on the boundary of~$f$.
  For a planar graph class~$\mathcal{G}$ we ask for the minimal number of edges needed to guard all faces of any~$n$-vertex graph in~$\mathcal{G}$.
  We prove that~$\lfloor n/3 \rfloor$ edges are always sufficient for quadrangulations and give a construction where~$\lfloor (n-2)/4 \rfloor$ edges are necessary.
  For~$2$-degenerate quadrangulations we improve this to a tight upper bound of~$\lfloor n/4 \rfloor$ edges.
  We further prove that~$\lfloor 2n/7 \rfloor$ edges are always sufficient for stacked triangulations (that are the $3$\nobreakdash-degenerate triangulations) and show that this is best possible up to a small additive constant.
\end{abstract}

\section{Introduction}

In 1975, Chv\'atal~\cite{Chvatal1975} laid the foundation for the widely studied field of \emph{art gallery problems} by answering how many guards are needed to observe all interior points of any given~$n$-sided polygon~$P$.
Here a guard is a point~$p$ in~$P$ and it can observe any other point~$q$ in~$P$ if the line segment~$pq$ is fully contained in~$P$.
He shows that~$\lfloor n/3 \rfloor$ guards are sometimes necessary and always sufficient.
Fisk~\cite{Fisk1978} revisited Chv\'atal's Theorem in 1978 and gave a very short and elegant new proof by introducing diagonals into the polygon~$P$ to obtain a triangulated, outerplanar graph.
Such graphs are~$3$-colorable and in each~$3$-coloring all faces are incident to vertices of all three colors, so the vertices of the smallest color class can be used as guard positions.
Countless variants of this problem have appeared in the literature, see the book by O'Rourke~\cite{ORourke1987} for a survey on several basic variants.
For example Kahn, Klawe and Kleitman~\cite{Kahn1983} prove that~$n$-vertex orthogonal polygons only require~$\lfloor n/4 \rfloor$ guards and several authors considered polygons with holes~\cite{Bjorling1995,Hoffmann1991}.
Bose et al.~\cite{Bose1997} studied the problem to guard the faces of a plane graph instead of a polygon.
A \emph{plane graph} is a graph~$G=(V,E)$ with an embedding in~$\mathbb{R}^2$ with not necessarily straight edges and without crossings between any two edges.
Here a face~$f$ is guarded by a vertex~$v$, if~$v$ is on the boundary of~$f$.
They show that~$\lfloor n/2 \rfloor$ vertices (so called \emph{vertex guards}) are sometimes necessary and always sufficient for~$n$-vertex plane graphs.

We consider a variant of this problem introduced by O'Rourke~\cite{ORourke1983}.
He shows that only~$\lfloor n/4 \rfloor$ guards are necessary in Chv\'atal's original setting if each guard is assigned to an edge of the polygon that he can patrol along instead of being fixed to a single point.
Considering plane graphs again, an \emph{edge guard} is an edge~$vw \in E$ and it guards all faces having~$v$ and/or~$w$ on their boundary.
For a given planar graph class~$\mathcal{G}$, we ask for the minimal number of edge guards needed to guard all faces of any~$n$-vertex graph in~$\mathcal{G}$.
Here and in the following the outer face is treated just like any other face and must also be guarded.

General (not necessarily triangulated)~$n$-vertex plane graphs might need at least~$\lfloor n/3 \rfloor$ edge guards, even when requiring~$2$-connectedness~\cite{Bose1997}.
The best known upper bounds have recently been presented by Biniaz et al.~\cite{Biniaz2019} and come in two different fashions:
First, any~$n$-vertex plane graph can be guarded by~$\lfloor 3n/8 \rfloor$ edge guards found in an iterative process.
Second, a coloring approach yields an upper bound of~$\lfloor n/3 + \alpha/9 \rfloor$ edge guards where~$\alpha$ counts the number of quadrangular faces in~$G$.
Looking at~$n$-vertex triangulations, Bose et al.~\cite{Bose1997} give a construction for triangulations needing~$\lfloor (4n-8)/13 \rfloor$ edge guards\footnote{
  The authors of~\cite{Bose1997} actually claim that~$\lfloor (4n-4)/13 \rfloor$ edge guards are necessary, but this result is only valid for \emph{near}-triangulations (this was noted first by Kau\v{c}i\v{c} et al.~\cite{Kaucic2003} and later clarified by one of the original authors~\cite{Bose2009}).
  For proper triangulations an additional vertex is needed in the construction so only~$\lfloor (4n-8)/13 \rfloor$ edge guards are necessary.
}.
A corresponding upper bound of~$\lfloor n/3 \rfloor$ edge guards was published earlier in the same year by Everett and Rivera-Campo~\cite{Everett1997}.

\paragraph{Preliminaries.}
All graphs considered throughout this paper are undirected and simple (unless explicitly stated otherwise).
Let~$G=(V,E)$ be a graph.
For an edge~$\{v,w\} \in E$ we use the shorter notation~$vw$ or~$wv$ and both mean the same.
The \emph{order} of~$G$ is its number of vertices and denoted by~$\lvert G \rvert$.
Similarly the \emph{size} of~$G$ is its number of edges.
We say that~$G$ is \emph{$k$\nobreakdash-regular} if each vertex~$v \in V$ has degree exactly~$k$.
Further~$G$ is called \emph{$k$\nobreakdash-degenerate} if every subgraph contains a vertex of degree at most~$k$.
For the subgraph induced by a subset~$X \subseteq V$ of the vertices we write~$G[X]$.
Now assume that~$G$ is plane with face set~$F$.
The \emph{dual graph}~$G^*=(V^*,E^*)$ is defined by $V^* = \{f^* \mid f \in F\}$ and~$E^* = \{f^*g^* \mid f,g \in F \land \text{$f,g$ share a boundary edge~$vw \in E$}\}$.
Note that~$G^*$ can be a multigraph.
The dual graph~$G^*$ of a plane graph~$G$ is also planar and we assume below that a plane drawing of~$G^*$ is given that is inherited from the plane drawing of~$G$ as follows:
Each dual vertex~$f^* \in V^*$ is drawn inside face~$f$, each dual edge~$f^*g^* \in E^*$ crosses its primal edge~$vw$ exactly once and in its interior and no two dual edges cross.

Let~$\Gamma \subseteq E$ be a set of edges.
We write~$V(\Gamma)$ for the set of endpoints of all edges in~$\Gamma$.
Further,~$\Gamma$ is an \emph{edge guard set} if all faces~$f \in F$ are guarded by at least one edge in~$\Gamma$, i.e. each face~$f$ has a boundary vertex in~$V(\Gamma)$.

\paragraph{Contribution.}
In Section~\ref{sec:quad} we consider the class of quadrangulations, i.e. plane graphs where every face is bounded by a~$4$-cycle.
We describe a coloring based approach to improve the currently best known upper bound to~$\lfloor n/3 \rfloor$ edge guards.
In addition we also consider~$2$-degenerate quadrangulations and present an upper bound of~$\lfloor n/4 \rfloor$ edge guards, which is best possible.
Our motivation to consider quadrangulations is that the coloring approaches developed earlier for general plane graphs~\cite{Biniaz2019,Everett1997} fail on quadrangular faces.
As a second result in Section~\ref{sec:stacked} we present a new upper bound of~$\lfloor 2n/7 \rfloor$ edge guards for stacked triangulations and show that it is best possible.
We saw above that no infinite family of triangulations is known that actually needs~$\lfloor n/3 \rfloor$ edge guards.
With the stacked triangulations we now know a non-trivial subclass of triangulations for which strictly fewer edge guards are necessary than for general plane graphs.
A short overview of previous and new results is given in Table~\ref{tbl:results}.

The proofs for all presented upper bounds here are constructive and we give additional details on how to turn them into efficient algorithms.

\begin{table}
  \centering
  \caption{
    Previous results (top) and contributions from this paper (bottom) for edge guard sets in plane graphs.
    Parameter~$\alpha$ counts the number of quadrangular faces.
  }
  \renewcommand{\arraystretch}{1.3}
  \begin{tabular}{lcc}
    \toprule
    Graph Class & Lower Bound & Upper Bound \\
    \midrule
    Plane Graph &
    $\left\lfloor \frac{n}{3} \right\rfloor$~\cite{Bose1997} &
    $\min\left\{\left\lfloor \frac{3n}{8} \right\rfloor, \left\lfloor \frac{n}{3} + \frac{\alpha}{9} \right\rfloor\right\}$~\cite{Biniaz2019} \\

    Triangulation &
    $\left\lfloor \frac{4n-8}{13} \right\rfloor$~\cite{Bose1997} &
    $\left\lfloor \frac{n}{3} \right\rfloor$~\cite{Everett1997} \\

    Outerplanar &
    $\left\lfloor \frac{n}{3} \right\rfloor$~\cite{Bose1997} &
    $\left\lfloor \frac{n}{3} \right\rfloor$~\cite{Chvatal1975} \\

    Maximal Outerplanar &
    $\left\lfloor \frac{n}{4} \right\rfloor$~\cite{ORourke1983} &
    $\left\lfloor \frac{n}{4} \right\rfloor$~\cite{ORourke1983} \\

    \midrule
    Stacked Triangulation &
    $\left\lfloor \frac{2n-4}{7} \right\rfloor$ &
    $\left\lfloor \frac{2n}{7} \right\rfloor$ \\

    Quadrangulation &
    $\left\lfloor \frac{n-2}{4} \right\rfloor$ &
    $\left\lfloor \frac{n}{3} \right\rfloor$ \\

    $2$-Degenerate Quadrangulation &
    $\left\lfloor \frac{n-2}{4} \right\rfloor$ &
    $\left\lfloor \frac{n}{4} \right\rfloor$ \\
    \bottomrule
  \end{tabular}
  \label{tbl:results}
\end{table}

\section{Quadrangulations}
\label{sec:quad}

Quadrangulations are the maximal plane bipartite graphs and every face is bounded by exactly four edges.
The currently best known upper bounds are the ones given by Biniaz et al.~\cite{Biniaz2019} for general plane graphs ($\lfloor 3n/8 \rfloor$ respectively~$\lfloor n/3 + \alpha/9 \rfloor$, where~$\alpha$ is the number of quadrilateral faces).
For~$n$\nobreakdash-vertex quadrangulations we have~$\alpha = n-2$, so~$\lfloor n/3 + (n-2)/9 \rfloor = \lfloor (4n-2)/9 \rfloor > \lfloor 3n/8 \rfloor$ for~$n \geq 4$.
In this section we provide a better upper bound of~$\lfloor n/3 \rfloor$ and a construction for quadrangulations needing~$\lfloor (n-2)/4 \rfloor$ edge guards.
Closing the gap remains an open problem.

\begin{theorem}
  \label{thm:quadsLowerBound}
  For~$k \in \mathbb{N}$ there exists a quadrangulation~$Q_k$ with~$n = 4k + 2$ vertices needing~$k = (n-2)/4$ edge guards.
\end{theorem}

\begin{proof}
  Define~$Q_k = (V,E)$ with~$V := \left\{s, t\right\} \cup \bigcup_{i = 1}^k \{a_i, b_i, c_i, d_i\}$ and~$E := \bigcup_{i=1}^k \{sa_i,\, sc_i,\, ta_i,\allowbreak tc_i,\, a_ib_i,\, a_id_i,\, c_ib_i,\, c_id_i\}$ as the union of $k$ disjoint $4$\nobreakdash-cycles and two extra vertices~$s$ and $t$ connecting them.
  Figure~\ref{fig:quadLowerBound} shows this and a planar embedding.
  Now for any two distinct~$i, j \in \{1,\dots,k\}$ the two quadrilateral faces~$(a_i, b_i, c_i, d_i)$ and~$(a_j, b_j, c_j, d_j)$ are only connected via paths through~$s$ or~$t$.
  Therefore, no edge can guard two or more of them and we need at least~$k$ edge guards for~$Q_k$.
  On the other hand it is easy to see that~$\{sa_1, \ldots, sa_k\}$ is an edge guard set of size~$k$, so~$Q_k$ needs exactly~$k$ edge guards.
\end{proof}

\begin{figure}
  \centering
  \begin{minipage}{0.45\textwidth}
    \centering
    \includegraphics{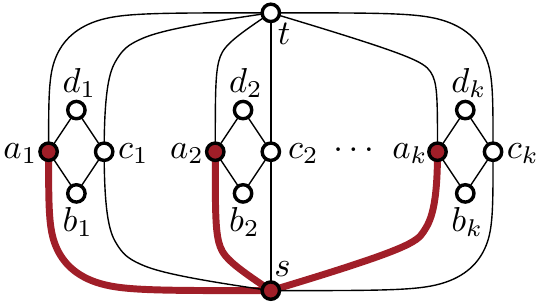}
    \caption{A quadrangulation with~$4k + 2$ vertices needing~$k$ edge guards (thick red edges).}
    \label{fig:quadLowerBound}
  \end{minipage}
  \hfill
  \begin{minipage}{0.45\textwidth}
    \centering
    \includegraphics{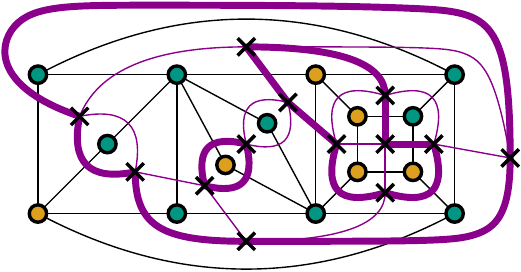}
    \caption{
      A quadrangulation~$G$ (black edges) and its dual~$G^*$ (purple edges) with a~$2$-factor (thick edges).
      The vertex coloring in orange and green is a guard coloring.
    }
    \label{fig:quadsUpperBound}
  \end{minipage}
\end{figure}

\noindent
The following definition and lemma are from Bose et al.~\cite{Bose2003} and we cite it using the terminology of~Biniaz et al.~\cite{Biniaz2019}.
A \emph{guard coloring} of a plane graph~$G$ is a non-proper $2$\nobreakdash-coloring of its vertex set, such that each face~$f$ of~$G$ has at least one boundary vertex of each color and at least one monochromatic edge (i.e. an edge where both endpoints receive the same color).
They prove that a guard coloring exists for all graphs without any quadrangular faces.

\begin{lemma}[{\cite[Lemma 3.1]{Bose2003}}]
  \label{lem:guardColoring}
  If there is a guard coloring for an~$n$-vertex plane graph~$G$, then~$G$ can be guarded by~$\lfloor n/3 \rfloor$ edge guards.
\end{lemma}

\noindent
Bose et al.~\cite{Bose2003} even present a linear time algorithm to compute a guard coloring for graphs without quadrangular faces.
We extend their result by showing that plane graphs consisting of only quadrangular faces also have a guard coloring.

\begin{theorem}
  \label{thm:quadsUpperBound}
  Every quadrangulation can be guarded by~$\lfloor n/3 \rfloor$ edge guards.
\end{theorem}

\begin{proof}
  Let~$G$ be a quadrangulation.
  We show that there is a guard coloring for~$G$, which is sufficient by Lemma~\ref{lem:guardColoring}.
  Consider the dual graph~$G^* = (V^*, E^*)$ of~$G$ with its inherited plane embedding, so each vertex~$f^* \in V^*$ is placed inside the face~$f$ of~$G$ corresponding to it.
  Since every face of~$G$ is bounded by a~$4$-cycle, its dual graph~$G^*$ is~$4$-regular.
  Using Petersen's $2$\nobreakdash-Factor Theorem~\cite{Petersen1891}\footnote{
    Diestel~\cite[Corollary~2.1.5]{Diestel2016} gives a very short and elegant proof of this theorem in his book.
    He only considers simple graphs there, but all steps in the proof also work for multigraphs like~$G^*$ that have at most two edges between any pair of vertices.
  } we get that~$G^*$ contains a~$2$-factor~$H$ (a spanning~$2$-regular subgraph).
  Therefore~$H$ is a set of vertex-disjoint cycles that might be nested inside each other.
  Now we define a $2$\nobreakdash-coloring~$\mathrm{col} : V \to \{0,1\}$ for the vertices of~$G$:
  For each~$v \in V$ let~$c_v$ be the number of cycles~$\mathcal{C}$ of~$H$ such that~$v$ belongs to the region of the embedding surrounded by~$\mathcal{C}$.
  The color of~$v$ is determined by the parity of~$c_v$ as~\mbox{$\mathrm{col}(v) := c_v \mod 2$}.

  We claim that this yields a guard coloring of~$G$:
  Any edge~$e = ab \in E$ has a corresponding dual edge~$e^*$.
  If~$e^* \in E(H)$, then~$e$ crosses exactly one cycle edge so~$\lvert c_a - c_b \rvert = 1$ and therefore~$\mathrm{col}(a) \neq \mathrm{col}(b)$.
  Otherwise~$e \not\in E(H)$ and its two endpoints are in the same cycles, thus~$\mathrm{col}(a) = \mathrm{col}(b)$ and~$e$ is monochromatic.
  Because~$H$ is a~$2$-factor, each face has exactly two monochromatic edges.
\end{proof}

\noindent
Figure~\ref{fig:quadsUpperBound} shows an example quadrangulation and a~$2$-factor in its dual graph.
By counting how many cycles each vertex lies inside, the vertices were colored in green and orange to obtain a guard coloring.
To complete this section we further note that above proof can be transformed into an efficient algorithm.

\begin{corollary}
  An edge guard set of size~$\lfloor n/3 \rfloor$ for an~$n$-vertex quadrangulation~$G$ can be computed in time~$O(n^{3/2})$.
\end{corollary}

\noindent
Bose et al.~\cite{Bose2003} already describe how to compute a guard set from a guard coloring in linear time, so it remains to check the time needed to obtain the guard coloring.
In his proof of Petersen's~$2$\nobreakdash-Factor Theorem Diestel~\cite{Diestel2016} reduces the problem to find a~$2$\nobreakdash-factor to finding a perfect matching in a bipartite and~$2$\nobreakdash-regular graph~$H$ with order and size linear in~$n$.
A perfect matching~$M$ in~$H$ exists by Hall's Theorem~\cite{Hall1935} and can be computed in time~$O(n^{3/2})$ using the Planar Separator Theorem~\cite{Lipton1980}.
The construction of~$H$ takes linear time and similarly a perfect matching~$M$ of~$H$ can be collapsed into a~$2$-factor of~$G$ in linear time.

In order to bridge the gap between the construction needing~$\lfloor (n-2)/4 \rfloor$ edge guards and the upper bound of~$\lfloor n/3 \rfloor$, we also consider the subclass of~$2$-degenerate quadrangulations in the master's thesis of the first author~\cite[Theorem 5.9]{Jungeblut2019}:

\begin{theorem}
  \label{thm:quads2DegUpperBound}
  Every~$n$-vertex~$2$-degenerate quadrangulation can be guarded by $\lfloor n/4 \rfloor$ edge guards.
\end{theorem}

\noindent
Note that this bound is best possible, as the quadrangulations constructed in Theorem~\ref{thm:quadsLowerBound} are~$2$-degenerate.
The proof of Theorem~\ref{thm:quads2DegUpperBound} follows the same lines as the one we present in the following section for stacked triangulations\footnote{
  For every~$n$-vertex $2$-degenerate quadrangulation~$G$ there is an~$(n - k)$-vertex $2$\nobreakdash-degenerate quadrangulation~$G'$ ($k \geq 4$), such that an edge guard set~$\Gamma'$ for~$G'$ can be used to construct an edge guard set~$\Gamma$ for~$G$ with~$\lvert \Gamma \rvert = \lvert \Gamma' \rvert + 1$.
  Obviously different cases need to be considered compared to stacked triangulations and analogous versions of Lemma~\ref{lem:stackedWeakForcingLemma} and Lemma~\ref{lem:stackedStrongForcingLemma} for quadrangular faces are needed.
  But apart from that the proof strategy is the same.
}.
Just as for stacked triangulations the proof can easily be turned into a linear time algorithm.

\section{Stacked Triangulations}
\label{sec:stacked}

The stacked triangulations (also known as Apollonian networks, maximal planar chordal graphs or planar~$3$\nobreakdash-trees) are a subclass of the triangulations that can recursively be formed by the following two rules:
\begin{enumerate*}[label=(\roman*)]
  \item A~triangle is a stacked triangulation and
  \item if~$G$ is a stacked triangulation and~$f = (x,y,z)$ an inner face, then the graph obtained by placing a new vertex~$v$ into~$f$ and connecting it with all three boundary vertices is again a stacked triangulation.
\end{enumerate*}
\begin{definition}
  For a stacked triangulation~$G$ we define~$\mathrm{height}(G)$ as
  \[
    \mathrm{height}(G) :=
    \begin{cases}
      0 & \text{if $\lvert G \rvert = 3$} \\
      1 + \max \{\mathrm{height}(G_1), \mathrm{height}(G_2), \mathrm{height}(G_3)\} & \text{otherwise}
    \end{cases}
  \]
  where~$G_1,G_2,G_3$ are the stacked triangulations induced by~$(v,x,y)$, $(v,y,z)$, $(v,z,x)$ and their interior vertices, respectively.
\end{definition}
\noindent
The stacked triangulations are a non-trivial subclass of the triangulations and we shall prove that they need strictly less than~$\lfloor n/3 \rfloor$ edge guards (which is the best known upper bound for general triangulations).
To start, we present a family of stacked triangulations needing many edge guards allowing us to conclude that the upper bound presented later is tight.

\begin{theorem}
  For even~$k \in \mathbb{N}$ there is a stacked triangulation~$G_k$ with~$n = (7k+4)/2$ vertices needing at least~$k = (2n - 4)/7$ edge guards.
\end{theorem}

\begin{proof}
  \begin{figure}
    \centering
    \includegraphics{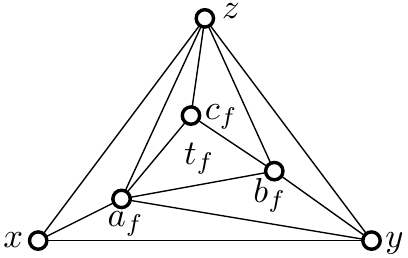}
    \caption{
      Three new vertices~$a_f,b_f,c_f$ are stacked into face~$f=(x,y,z)$ forming a new face~$t_f$.
      Note that the graph remains a stacked triangulation.
    }
    \label{fig:stackedLowerBound}
  \end{figure}
  Let~$S$ be a stacked triangulation with~$k$ faces and therefore~$(k+4)/2$ vertices (by Euler's formula).
  Insert three new vertices~$a_f,b_f,c_f$ into each face~$f$ of~$S$ such that the resulting graph is still a stacked triangulation and these three vertices form a new triangular face~$t_f$, i.e.~$f$ and~$t_f$ do not share any boundary vertices.
  Figure~\ref{fig:stackedLowerBound} illustrates how the new vertices can be inserted into a single face~$f$.
  Then~$G$ has~$n = (k+4)/2 + 3k = (7k+4)/2$ vertices.
  For any two distinct faces~$f,g$ of~$S$ let~$P$ be a shortest path between any two boundary vertices of the new faces~$t_f$ and~$t_g$.
  By our construction~$P$ has length at least~$2$, so no edge can guard both~$t_f$ and~$t_g$.
  Therefore~$G$ needs at least~$k$ edge guards.
\end{proof}

\noindent
Complementing this construction we state the following upper bound.

\begin{theorem}
  \label{thm:stackedUpperBound}
  Every~$n$-vertex stacked triangulation with~$n \geq 4$ can be guarded by~$\lfloor 2n/7 \rfloor$ edge guards.
\end{theorem}

\noindent
Before going into a detailed proof, let us start with a high-level description of the proof strategy.
We use induction on the number~$n$ of vertices.
Given a stacked triangulation~$G$ we create a smaller stacked triangulation~$G'$ of size~$\lvert G' \rvert = \lvert G \rvert - k$ for some~$k \in \mathbb{N}$.
Applying the induction hypothesis on~$G'$ yields an edge guard set~$\Gamma'$ of size~$\lvert \Gamma' \rvert \leq \lfloor 2(n-k)/7\rfloor$.
Then we extend~$\Gamma'$ into an edge guard set~$\Gamma$ for~$G$ using~$\ell$ additional edges.
In each step we guarantee~$\ell/k \leq 2/7$, such that~$\lvert \Gamma \rvert = \lvert \Gamma' \rvert + \ell \leq \lfloor 2(n-k)/7\rfloor + 2k/7 = \lfloor 2n/7 \rfloor$.

We create~$G'$ by choosing a triangle~$\bigtriangleup$ and removing the set of vertices in its interior.
Call this set~$V^-$.
Note that~$G'$ is still a stacked triangulation.
Under all possible candidates we choose~$\bigtriangleup$ such that~$V^-$ is of minimal cardinality but consists of at least four elements.
By the choice of~$\bigtriangleup$ we get that~$G[\bigtriangleup \cup V^-]$ has at most ten inner vertices:
The triangle~$\bigtriangleup$ consists of three vertices~$x,y,z$ and there is a unique vertex~$v$ in its interior adjacent to all three of them.
The remaining vertices of~$G[\bigtriangleup \cup V^-]$ are distributed along the three triangles~$(v,x,y)$,~$(v,y,z)$ and~$(v,z,x)$.
None of them can contain more than three vertices in its interior, otherwise it would be a triangle~$\bigtriangleup'$ that would have been chosen instead of~$\bigtriangleup$.

Now that we have a bound on the size of~$G[\bigtriangleup \cup V^-]$, we systematically consider edge guard sets for small stacked triangulations.
This requires to consider many cases and the work is distributed among several observations and lemmas:
Observations~\ref{obs:doublyTriply} and~\ref{obs:stacked3Wheel} as well as Lemmas~\ref{lem:stacked6} and~\ref{lem:stacked7} are key insights about small stacked triangulations and are all used several times later on.
On the other hand, Lemmas~\ref{lem:stacked2Vertex5},~\ref{lem:stacked2Vertex6} and~\ref{lem:stacked2Vertex7} are special cases emerging in the proof of Theorem~\ref{thm:stackedUpperBound} and are proved in isolation beforehand for a simpler and more modular chain of arguments.
These three lemmas themselves are based on techniques to restrict the edge guard set~$\Gamma'$ given by the induction hypothesis which are developed in Lemmas~\ref{lem:stackedWeakForcingLemma} and~\ref{lem:stackedStrongForcingLemma}.
All this groundwork is then ultimately combined into a relatively simple proof of Theorem~\ref{thm:stackedUpperBound}.

\begin{observation}
  \label{obs:doublyTriply}
  Let~$f = (x,y,z)$ be a face of a stacked triangulation and let~$\Gamma$ be an edge guard set.
  Now we add a new vertex~$v$ into~$f$ with edges to all~$x,y,z$.
  \begin{itemize}
    \item If~$\lvert V(\Gamma) \cap \{x,y,z\} \rvert \geq 2$, we say that~$f$ is \emph{doubly} guarded.
    In this case the three new faces~$(v,x,y)$, $(v,y,z)$ and~$(v,z,x)$ are all guarded.
    \item If~$\lvert V(\Gamma) \cap \{x,y,z\} \rvert = 3$, we say that~$f$ is \emph{triply} guarded.
    Furthermore the three new faces~$(v,x,y)$, $(v,y,z)$ and~$(v,z,x)$ are all doubly guarded.
  \end{itemize}
\end{observation}

\begin{observation}
  \label{obs:stacked3Wheel}
  Let~$G$ be a stacked triangulation, and~$v$ be a vertex of degree~$3$ with neighbors~$x,y,z$.
  Then for any edge guard set~$\Gamma$ we have~$\lvert \{v,x,y,z\} \cap V(\Gamma) \rvert \geq 2$.
  If~$v \not\in V(\Gamma)$, then at least two of~$x,y,z$ must be in~$V(\Gamma)$, because each of them is incident to only two of the three faces inside~$(x,y,z)$.
  But if~$v \in V(\Gamma)$, it must be as part of an edge with one of its neighbors.
\end{observation}

\begin{lemma}
  \label{lem:stacked6}
  Let~$G$ be a~$6$-vertex stacked triangulation.
  Then~$G$ can be guarded by a unique edge guard.
  Further, if there is a vertex guard at an outer vertex~$x$ of~$G$, then there is an edge~$ab$ guarding the remaining faces where~$a \neq x$ is another outer vertex.
\end{lemma}

\begin{proof}
  There is only a single~$6$-vertex stacked triangulation and it has three substantially different planar embeddings, all shown in Figure~\ref{fig:stacked6}.
  We see that there is indeed only one possible edge guard~$vw$ in all three cases.
  Further both~$v$ and~$w$ are adjacent to all three outer vertices (or are one of them and adjacent to the other two).
  \begin{figure}
    \centering
    \includegraphics{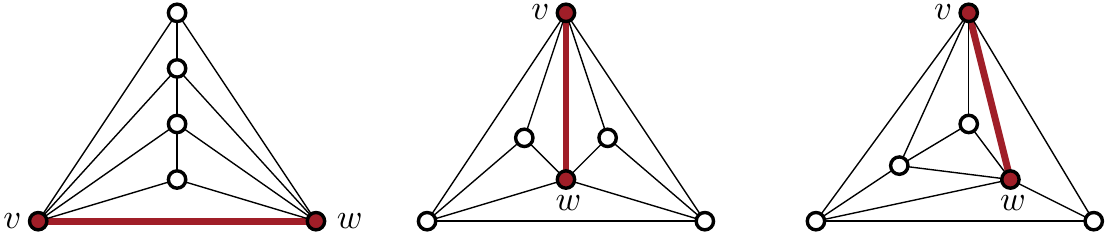}
    \caption{
      The three different planar embeddings of the unique~$6$-vertex stacked triangulation.
      The thick red edge is the unique edge guard for all eight faces.
    }
    \label{fig:stacked6}
  \end{figure}
\end{proof}

\begin{lemma}
  \label{lem:stacked7}
  Let~$G$ be a~$7$-vertex stacked triangulation with a vertex guard at an outer vertex.
  Then one additional edge suffices to guard the remaining faces of~$G$.
\end{lemma}

\begin{proof}
  Let~$(x,y,z)$ be the outer face of~$G$ and~$v$ be the unique vertex adjacent to all three of them.
  We distinguish the three different cases shown in Figure~\ref{fig:stacked7} on how the remaining three vertices are distributed.
  \begin{figure}
    \centering
    \includegraphics{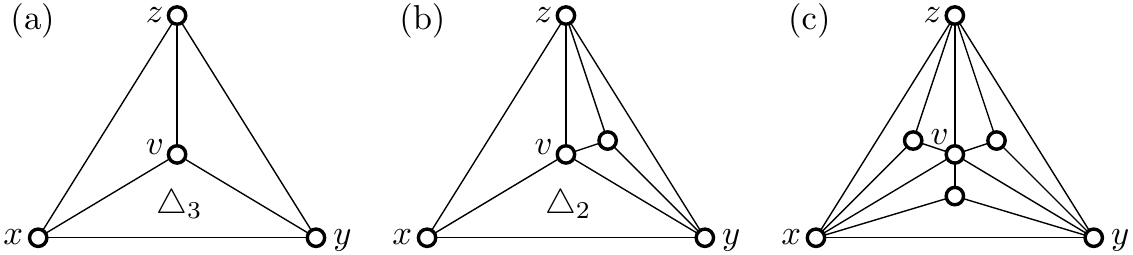}
    \caption{
      The three ways how a~$7$-vertex stacked triangulation can look like.
      By~$\bigtriangleup_k$ we denote all possibilities to insert~$k$ vertices such that the graph is a stacked triangulation.
    }
    \label{fig:stacked7}
  \end{figure}

  \begin{description}
    \item[Case (a):]
    If~$z$ is given as a vertex guard, the remaining faces can be guarded by the unique edge guarding the induced~$6$-vertex stacked triangulation bounded by~$(v,x,y)$.
    If otherwise without loss of generality~$x$ is given as a vertex guard, there is an edge that guards the remaining faces inside triangle~$(v,x,y)$ and contains~$v$ or~$y$.
    In both cases this edge exists by Lemma~\ref{lem:stacked6}.

    \item [Case (b):]
    If~$x$ is guarded, we can use edge~$vy$ and if~$y$ is guarded we can use edge~$vx$.
    In both cases, triangle~$(v,y,z)$ is doubly guarded and~$(v,x,y)$ is triply guarded, so all interior faces are guarded.
    In the remaining case that~$z$ is guarded, there is a unique edge of triangle~$(v,x,y)$ guarding all faces inside it and also containing at least one of~$v$ and~$y$, so triangle~$(v,y,z)$ is doubly guarded.
    Then all interior faces are guarded.

    \item [Case (c):]
    For a guarded outer vertex choose the edge connecting the other two outer vertices.
    All inner faces are incident to at least one of the three outer vertices, so they are all guarded.
  \end{description}
\end{proof}

\noindent
We can already see how Lemma~\ref{lem:stacked7} is used in our inductive step, namely in all cases where~$\lvert V^- \rvert = 4$:
After removing the vertices in~$V^-$ the triangle~$\bigtriangleup$ from~$G$ is a face in~$G'$ and this face gets guarded by any edge guard set~$\Gamma'$ for~$G'$.
Using Lemma~\ref{lem:stacked7} we know that one additional edge is always enough to extend~$\Gamma'$ to an edge guard set~$\Gamma$ for~$G$.
However, for~$\lvert V^- \rvert \geq 5$ the situation gets more complex; just removing~$V^-$ and applying the induction hypothesis might lead to an edge guard set~$\Gamma'$ for~$G'$ that cannot be extended to an edge guard set~$\Gamma$ for~$G$ with~$\ell \leq 2k/7$ additional edges (remember that~$k = \lvert G \rvert - \lvert G' \rvert)$.
See Figure~\ref{fig:stackedNeedsForcing} as an example.
To solve this and similar cases we now describe two ways how to extend~$G'$ with some new vertices and edges, such that there is always at least some edge guard set~$\Gamma'$ for~$G'$ of size~$\lfloor 2\lvert G' \rvert / 7\rfloor$ that can be augmented into an edge guard set~$\Gamma$ for~$G$ of size~$\lfloor 2\lvert G \rvert / 7 \rfloor$.

\begin{figure}
  \centering
  \includegraphics{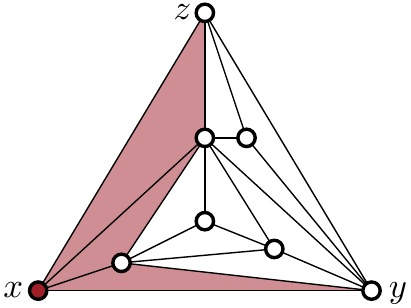}
  \caption{
    Here~$\bigtriangleup = (x,y,z)$ and~$G'$ was created by removing all interior vertices~$V^-$.
    The induction hypothesis then provided an edge guard set~$\Gamma'$ for~$G'$ that guards face~$(x,y,z)$ of~$G'$ through~$x \in V(\Gamma')$.
    After reinserting~$V^-$, the faces shaded in red are already guarded.
    But none of the other edges is strong enough to guard the remaining faces.
  }
  \label{fig:stackedNeedsForcing}
\end{figure}

\begin{lemma}
  \label{lem:stackedWeakForcingLemma}
  Let~$f = (x,y,z)$ be a face of a stacked triangulation.
  By adding two new vertices into~$f$ we can obtain a stacked triangulation~$G$ such that for each edge guard set~$\Gamma$ there is an edge guard set~$\Gamma'$ of equal size with~$\{x,y\} \subseteq V(\Gamma')$.
\end{lemma}

\begin{proof}
  \begin{figure}
    \centering
    \includegraphics{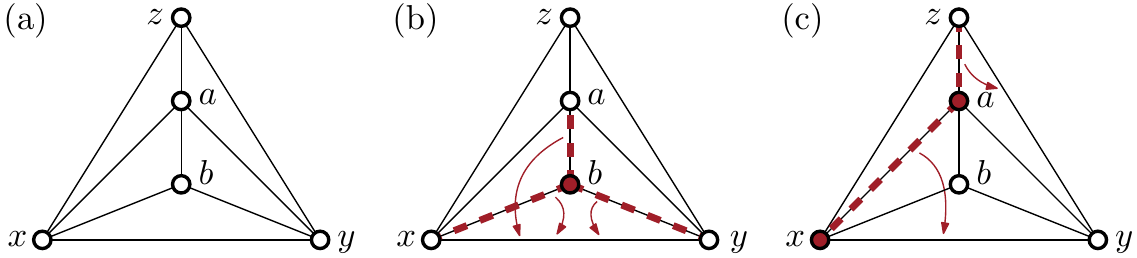}
    \caption{
      (a)~Add~$a,b$ as shown to get~$x,y \subseteq V(\Gamma)$.
      (b)~The thick dashed edges are all possible edge guards containing~$b$ as an endpoint.
      The little arrows indicate that each edge can be exchanged with edge~$xy$.
      (c)~The thick dashed edges are all possible edge guards containing~$a$ as an endpoint if~$b,y \not\in V(\Gamma)$.
      Again the little arrows indicate with which edge these edge guards can be exchanged.
    }
    \label{fig:stackedWeakForcingLemma}
  \end{figure}
  Add vertex~$a$ with edges~$ax, ay, az$ and then vertex~$b$ with edges~$ab, bx, by$ to obtain~$G$ as shown in Figure~\ref{fig:stackedWeakForcingLemma}a.
  Now let~$\Gamma$ be an edge guard set for~$G$ with $\lvert \{x,y\} \cap V(\Gamma) \rvert \leq 1$.
  If~$b \in V(\Gamma)$ as part of an edge~$bv$, we can set~$\Gamma' := (\Gamma \setminus \{bv\}) \cup \{xy\}$, see Figure~\ref{fig:stackedWeakForcingLemma}b.
  This is possible, because no matter what vertex~$v$ is, edge~$xy$ guards a superset of the faces that~$bv$ guards.
  If otherwise~$b \not\in V(\Gamma)$, we assume without loss of generality that~$x \in V(\Gamma)$ so that face~$(x,y,b)$ is guarded.
  Face~$(a,b,y)$ can then only be guarded by edge~$av$ where~$v \in \{x,z\}$.
  Since~$N(a) \subseteq N(y)$ we can set~$\Gamma' := (\Gamma \setminus \{av\}) \cup \{vy\}$, see Figure~\ref{fig:stackedWeakForcingLemma}c.
  In both cases~$\{x,y\} \subseteq \Gamma'$ and~$\lvert \Gamma \rvert = \lvert \Gamma' \rvert$.
\end{proof}

\begin{lemma}
  \label{lem:stackedStrongForcingLemma}
  Let~$(x,y,z)$ be a face of a stacked triangulation.
  By adding three new vertices~$a,b,c$ into~$(x,y,z)$ we can obtain a stacked triangulation~$G$ such that for each edge guard set~$\Gamma$ of~$G$ there is an edge guard set~$\Gamma'$ of equal size with~$x \in V(\Gamma')$ and an edge~$vw \in \Gamma'$ with~$v \in \{x,y,z\}$ and~$w \in \{a,b,c\}$.
\end{lemma}

\begin{proof}
  \begin{figure}
    \centering
    \includegraphics{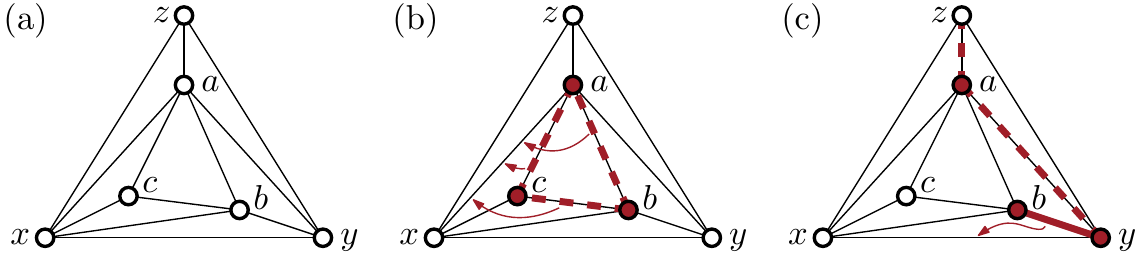}
    \caption{
      (a)~Add~$a,b,c$ like this into face~$(x,y,z)$.
      (b)~An edge guard with both endpoints in~$\{a,b,c\}$ can be exchanged by an edge guard~$ax$.
      (c)~If~$by$ and vertex~$a$ are guarded, exchange~$by$ by~$xy$.
    }
    \label{fig:stackedStrongForcingLemma}
  \end{figure}
  Add vertex~$a$ with edges~$ax, ay, az$, then vertex~$b$ with edges~$ab, bx, by$ and then vertex~$c$ with edges~$ac, bc, cx$ to obtain~$G$, see Figure~\ref{fig:stackedStrongForcingLemma}a.
  Now let~$\Gamma$ be an edge guard set for~$G$ that does not fulfill the requirements.
  If there is an edge~$uw$ with~$\{u,w\} \subseteq \{a,b,c\}$ we can set~$\Gamma' := (\Gamma \setminus \{uw\}) \cup \{ax\}$.
  This is because edge~$ax$ guards a superset of the faces that~$bv$ guards, see Figure~\ref{fig:stackedStrongForcingLemma}b.

  If such an edge does not exist, there must be some other edge~$vw \in \Gamma$ guarding face~$(a,b,c)$ with~$w \in \{a,b,c\}$ and~$v \in \{y,z\}$.
  Further, as~$x \not\in V(\Gamma)$ we also have~$c \not\in V(\Gamma)$, otherwise the previous case would apply.
  This leaves vertex~$b$ to be the only possible vertex to guard face~$(b,c,x)$ and vertex~$a$ to be the only possible vertex to guard face~$(a,x,c)$.
  The edge guard starting at~$b$ can only have~$y$ as its other endpoint, see Figure~\ref{fig:stackedStrongForcingLemma}c.
  In any case we set~$\Gamma' := (\Gamma \setminus \{by\}) \cup \{xy\}$.
\end{proof}

\noindent
With Lemma~\ref{lem:stackedWeakForcingLemma} and Lemma~\ref{lem:stackedStrongForcingLemma} at hand we can now consider the cases where~$\lvert V^- \rvert \geq 5$, i.e. stacked triangulations on eight or more vertices.

\begin{figure}
  \centering
  \includegraphics{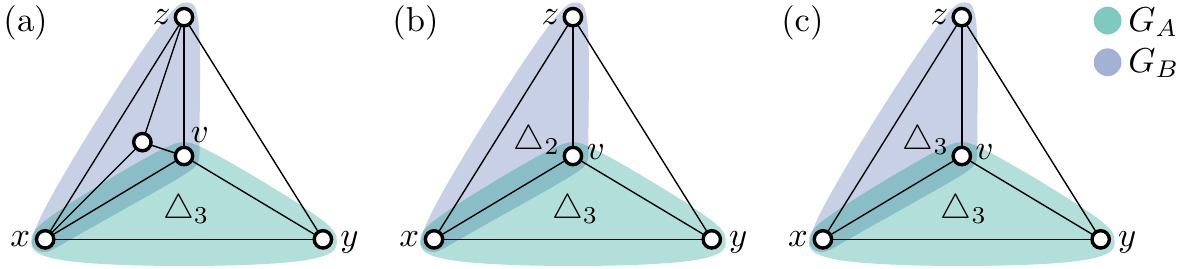}
  \caption{
    Configurations as described in
    (a)~Lemma~\ref{lem:stacked2Vertex5},
    (b)~Lemma~\ref{lem:stacked2Vertex6} and
    (c)~Lemma~\ref{lem:stacked2Vertex7}.
    Here~$\bigtriangleup_k$ is a placeholder for~$k$ additional vertices inside the surrounding triangle (such that the graph is a stacked triangulation).
    The subgraphs induced by the vertices highlighted in green (blue) induce the stacked triangulation~$G_A$ ($G_B$).
  }
  \label{fig:stacked2VertexOverview}
\end{figure}

\begin{lemma}
  \label{lem:stacked2Vertex5}
  Let~$G$ be an $8$-vertex stacked triangulation with outer face~$(x,y,z)$, such that the following configuration applies (see Figure~\ref{fig:stacked2VertexOverview}a):
  \begin{itemize}
    \item Vertex~$v$ is the only vertex adjacent to all~$x,y,z$.
    \item $(v,x,y)$ and its interior vertices induce a~$6$-vertex stacked triangulation~$G_A$.
    \item $(v,z,x)$ and its interior vertex induce a~$4$-vertex stacked triangulation~$G_B$.
  \end{itemize}
  Then any edge guard set~$\Gamma'$ for the subgraph~$G'$ induced by~$\{v,x,y,z\}$ can be extended by one edge~$e$ to an edge guard set~$\Gamma$ for~$G$.
\end{lemma}

\begin{proof}
  By Observation~\ref{obs:stacked3Wheel} we have~$\lvert \{v,x,y,z\} \cap V(\Gamma') \rvert \geq 2$.
  Face~$(v,y,z)$ is then already guarded.
  If further~$\lvert \{v,x,z\} \cap V(\Gamma') \rvert \geq 2$, then triangle~$(v,z,x)$ is doubly guarded, so all faces of~$G_B$ are guarded.
  In this case set~$e$ to be the unique edge guarding~$G_A$, which exists by Lemma~\ref{lem:stacked6}.

  If otherwise~$\lvert \{v,x,z\} \cap V(\Gamma') \rvert = 1$, we have~$y \in V(\Gamma')$.
  By Lemma~\ref{lem:stacked6} an edge~$ab$ exists that guards the remaining faces of~$G_A$ with~$a \in \{x,v\}$ but also $a \not\in V(\Gamma')$.
  Then~$G_B$ is doubly guarded so all of its faces are guarded.
\end{proof}

\begin{lemma}
  \label{lem:stacked2Vertex6}
  Let~$G$ be a~$9$-vertex stacked triangulation with outer face~$(x,y,z)$, such that the following configuration applies (see Figure~\ref{fig:stacked2VertexOverview}b):
  \begin{itemize}
    \item Vertex~$v$ is the only vertex adjacent to all~$x,y,z$.
    \item $(v,x,y)$ and its interior vertices induce a~$6$-vertex stacked triangulation~$G_A$.
    \item $(v,z,x)$ and its interior vertices induce a~$5$-vertex stacked triangulation~$G_B$.
  \end{itemize}
  Then we can create a~$5$-vertex stacked triangulation~$G'$, such that any edge guard set~$\Gamma'$ for~$G'$ can be augmented into an edge guard set~$\Gamma$ for~$G$ with~$\lvert \Gamma \rvert = \lvert \Gamma' \rvert + 1$.
\end{lemma}

\begin{proof}
  Note that the~$5$-vertex stacked triangulation~$G_B$ can always be guarded by one of its outer edges~$vx$, $vz$ or $xz$.
  In all cases we first remove the interior vertices of~$(x,y,z)$ and add two new vertices~$a$ and~$b$ into~$(x,y,z)$ to get a stacked triangulation~$G'$.
  By placing~$a$ and~$b$ appropriately, Lemma~\ref{lem:stackedWeakForcingLemma} allows us to force one of the three sets~$\{x,y\}$, $\{x,z\}$, $\{y,z\}$ to be a subset of~$V(\Gamma')$.
  Depending on which edge from~$\{vx, vz, xz\}$ guards~$G_B$, we force a different one of the three sets.
  In the following let~$e = uw$ be the unique edge guarding~$G_A$ such that~$u \in \{v,x,y\}$ by Lemma~\ref{lem:stacked6}.

  \begin{enumerate}[label=\textbf{Case \arabic*:}, ref=\arabic*, leftmargin=*, labelindent=0em, itemindent=3em]
    \item $xz$ guards~$G_B$: \\
    Place~$a$ and~$b$ such that~$x,z \in V(\Gamma')$.
    Then all faces of~$G_B$ and face~$(v,y,z)$ are guarded.
    We set~$\Gamma := \Gamma' \cup \{e\}$ to also guard all faces of~$G_A$.

    \item $xv$ guards~$G_B$: \\
    Place~$a$ and~$b$ such that~$x,y \in V(\Gamma')$.
    Then face~$(v,y,z)$ is already guarded.
    If~$u = v$, set~$\Gamma := \Gamma' \cup \{e\}$ to guard all faces of~$G_B$ and $G_A$.
    Otherwise we can use~$e' := vw$ instead by Lemma~\ref{lem:stacked6} and set~$\Gamma := \Gamma' \cup \{e'\}$.

    \item $vz$ guards~$G_B$: \\
    If~$u = v$, we can place~$a$ and~$b$ to force~$y,z \in V(\Gamma')$.
    Then~$\Gamma := \Gamma' \cup \{e\}$ guards all faces of~$G_A$ and~$G_B$.
    Otherwise, place~$a$ and~$b$ so that~$u,z \in V(\Gamma')$ and set~$e' := vw$ by Lemma~\ref{lem:stacked6}.
    Then~$\Gamma := \Gamma' \cup \{e'\}$ fulfills the requirements.
    \qedhere
  \end{enumerate}
\end{proof}

\begin{lemma}
  \label{lem:stacked2Vertex7}
  Let~$G$ be a~$10$-vertex stacked triangulation with outer face~$(x,y,z)$, such that the following configuration applies (see Figure~\ref{fig:stacked2VertexOverview}c):
  \begin{itemize}
    \item Vertex~$v \in V$ is the unique vertex adjacent to all~$x$,~$y$ and~$z$.
    \item $(v,x,y)$ and its interior vertices induce a~$6$-vertex stacked triangulation~$G_A$.
    \item $(v,z,x)$ and its interior vertices induce a~$6$-vertex stacked triangulation~$G_B$.
  \end{itemize}
  Then we can create a~$6$-vertex stacked triangulation~$G'$, such that any edge guard set~$\Gamma'$ for~$G'$ can be augmented to an edge guard set~$\Gamma$ for~$G$ with~$\lvert\Gamma\rvert = \lvert\Gamma'\rvert + 1$.
\end{lemma}

\begin{proof}
  We construct~$G'$ by removing all interior vertices of~$(x,y,z)$ and by adding three new vertices~$a,b,c$ such that~$G'$ is a stacked triangulation.
  Lemma~\ref{lem:stackedStrongForcingLemma} allows us to connect~$a,b,c$ in such a way that there is an edge guard set~$\Gamma'$ with a predefined vertex from~$\{x,y,z\}$ in~$V(\Gamma')$ and an edge~$e = uw \in \Gamma'$ with~$u \in \{x,y,z\}$ and~$w \in \{a,b,c\}$.
  When applying~$\Gamma'$ to~$G$, edge guard~$uw \in \Gamma'$ degenerates into a vertex guard at~$u$ and we can extend it arbitrarily to an edge guard~$uw' \in E$.

  Let us name two special edges of~$G$:
  Edge~$u_Aw_A$ is the unique edge guard of~$G_A$ and~$u_Bw_B$ is the unique edge guard of~$G_B$ by Lemma~\ref{lem:stacked6}.
  Let the naming be so that~$u_A$ and~$u_B$ are incident to the outer face of~$G_A$ and~$G_B$, respectively.
  Depending on which vertices of~$G$ the vertices~$u_A$ and~$u_B$ correspond to, we connect~$a,b,c$  differently to~$x,y,z$ to force a different vertex from~$\{x,y,z\}$ to be in~$\Gamma'$.
  In any case we then show which edge~$uw'$ to use instead of edge~$uw$ and which other edge to add to get edge guard set~$\Gamma$ for~$G$.

  \begin{enumerate}[label=\textbf{Case \arabic*:}, ref=\arabic*, leftmargin=*, labelindent=0em, itemindent=3em]
    \item $x = u_A = u_B$: \\
    Force vertex~$x \in V(\Gamma')$.
    If~$u = x$, choose~$uw' = xw_A$.
    If otherwise~$u \neq x$, we assume without loss of generality that~$u = y$ (the other case~$u = z$ works symmetrically) and choose~$uw' = yw_A$.
    In both cases set $\Gamma := (\Gamma' \setminus \{uw\}) \cup \{uw', u_Bw_B\}$.

    \item\label{itm:stacked2Vertex7Case2} $x = u_A \neq u_B$: ($x = u_B \neq u_A$ can be handled symmetrically) \\
    Force vertex~$x \in V(\Gamma')$.
    Now if~$u = x$, use choose~$uw' = xw_A$ and if~$u = y$, choose~$uw' = yw_A$.
    In both cases set $\Gamma := (\Gamma' \setminus \{uw\}) \cup \{uw', u_Bw_B\}$.
    If otherwise~$u = z$, choose~$uw' = zw_B$ and set $\Gamma := (\Gamma' \setminus \{uw\}) \cup \{uw', vw_A\}$.

    \item $v = u_A = u_B$: \\
    Force vertex~$x \in V(\Gamma')$.
    Without loss of generality we get~$u = x$ or~$u = y$.
    In both cases use~$uw' = uw_A$ and set $\Gamma := (\Gamma' \setminus \{uw\}) \cup \{uw', vw_B\}$.

    \item $y = u_A$: ($z = u_B$ can be handled symmetrically) \\
    Force vertex~$y \in V(\Gamma')$.
    If we have~$u = y$ or~$u = x$, use~$uw' = uw_A$ and set $\Gamma := (\Gamma' \setminus \{uw\}) \cup \{uw', u_Bw_B\}$.
    If otherwise~$u = z$, use~$uw' = zw_B$ and set $\Gamma := (\Gamma' \setminus \{uw\}) \cup \{uw', vw_A\}$ additional edge~$w_Av$.
    This always works, because~$u_B \in \{v,z\}$ since otherwise Case~\ref{itm:stacked2Vertex7Case2} applies.
  \end{enumerate}
  In all cases~$\{u_A,w_A,u_B,w_B\} \subseteq V(\Gamma)$, so all faces of~$G_A$ and~$G_B$ are guarded.
  If face~$(v,y,z)$ is also guarded by~$\Gamma$, we are done, so assume it is not.
  Then it must be~$x = u_A = u_B$, because both~$u_A$ and~$u_B$ are outer vertices of~$G_A$ and~$G_B$, respectively.
  We can change the edge containing~$u_B$ to end at~$v$ instead.
\end{proof}

\noindent
Finally we are set up to prove Theorem~\ref{thm:stackedUpperBound} stating that~$\lfloor 2n/7 \rfloor$ edge guards are always sufficient for any~$n$-vertex stacked triangulation~$G$.

\begin{proof}[Proof of Theorem~\ref{thm:stackedUpperBound}]
  As described above, the proof is by induction on the number~$n$ of vertices.
  We find a smaller graph~$G'$ for which the induction hypothesis provides an edge guard set~$\Gamma'$ that we augment into an edge guard set~$\Gamma$ for~$G$.
  By guaranteeing that~$(\lvert\Gamma\rvert - \lvert\Gamma'\rvert) / (\lvert G \rvert - \lvert G' \rvert) \leq 2/7$ we hereby obtain an edge guard set for~$G$ of size at most~$\lfloor 2n/7 \rfloor$.

  For base case we note that if~$n \leq 6$, we need a single edge guard by Lemma~\ref{lem:stacked6}.
  So from now on assume~$n \geq 7$.
  Let~$\bigtriangleup = (x,y,z)$ be a triangle such that there are at least four vertices~$V^-$ inside~$\bigtriangleup$ but among all candidates~$\lvert V^- \rvert$ is minimal.
  Further let~$v \in V^-$ be the unique vertex adjacent to all~$x,y,z$.
  We consider the following cases in the order they are given:
  If a case applies, then all others before must not apply.
  \begin{enumerate}[label=\textbf{Case \arabic*:}, ref=\arabic*, leftmargin=*, labelindent=0em, itemindent=3em]
    \item\label{itm:stacked4VertexSubtree}
    $\lvert V^- \rvert = 4$: \\
    Set~$G' := G[V \setminus V^-]$ and use the induction hypothesis to get an arbitrary edge guard set~$\Gamma'$ for~$G'$.
    Triangle~$\bigtriangleup$ is a face in~$G'$ and as such guarded by~$\Gamma'$ through at least one of its boundary vertices.
    Together with the vertices in~$V^-$ it forms a~$7$-vertex stacked triangulation with at least one guarded outer vertex, so by Lemma~\ref{lem:stacked7} we can extend~$\Gamma'$ by one additional edge to an edge guard set~$\Gamma$ for~$G$.
    We get~$k = \lvert G \rvert - \lvert G' \rvert = 4$ and~$\ell = 1$, so~$\ell / k = 1/4 \leq 2/7$.

    \item\label{itm:stackedMaxHeight3}
    $\lvert V^- \rvert \geq 5 \land \mathrm{height}(G[\bigtriangleup \cup V^-]) \leq 3$: \\
    Construct~$G'$ by removing all vertices from~$V^-$ except for~$v$ and let~$\Gamma'$ be an edge guard set for~$G'$ given by the induction hypothesis.
    By Observation~\ref{obs:stacked3Wheel} we have~$\lvert \{v,x,y,z\} \cap V(\Gamma') \rvert \geq 2$ and we set~$\Gamma$ to be~$\Gamma'$ plus one additional edge, so that~$\{v,x,y,z\} \subseteq V(\Gamma)$.
    This is always possible, because~$v,x,y,z$ induce a~$4$-clique in~$G$.
    Because~$\mathrm{height}(G[\bigtriangleup \cup V^-]) \leq 3$, each face of~$G$ inside~$\bigtriangleup$ is incident to at least one vertex in~$\{v,x,y,z\}$, so all faces are guarded.
    We get~$k = \lvert G \rvert - \lvert G' \rvert \geq 4$ and~$\ell = 1$, so~$\ell / k \leq 1/4 \leq 2/7$.
  \end{enumerate}

  \noindent
  At this stage we finished all cases where~$\mathrm{height}(G[\bigtriangleup \cup V^-]) \leq 3$.
  The following cases all have~$\mathrm{height}(G[\bigtriangleup \cup V^-]) = 4$.
  (Note \emph{equality} instead of \emph{greater than or equal}.
  This is justified, because if one of the three triangles~$(v,x,y)$, $(v,y,z)$ or $(v,z,x)$ has height at least five, it would contain at least four vertices in its interior.
  This is impossible as~$\bigtriangleup$ has a minimal number of vertices in its interior.)

  \begin{enumerate}[resume, label=\textbf{Case \arabic*:}, ref=\arabic*, leftmargin=*, labelindent=0em, itemindent=3em]
    \item\label{itm:stackedHeight4_2vertex}
    $\mathrm{height}(G[\bigtriangleup \cup V^-]) = 4 \land [\text{$(v,x,y)$, $(v,y,z)$ or $(v,z,x)$ is a face}]$: \\
    Without loss of generality we assume that~$(v,y,z)$ is a face.
    The other cases are symmetric.
    If~$\lvert V^- \rvert = 5/6/7$, then~$G[\bigtriangleup \cup V^-]$ induces an~$8/9/10$-vertex stacked triangulation fulfilling the conditions of Lemma~\ref{lem:stacked2Vertex5}/\ref{lem:stacked2Vertex6}/\ref{lem:stacked2Vertex7}, respectively.
    In all three cases, the lemma describes how~$G'$ is constructed and how an edge guard set~$\Gamma'$ obtained by applying the induction hypothesis can be extended to~$\Gamma$.
    We always have~$k = \lvert G \rvert - \lvert G' \rvert \geq 4$ and~$\ell = 1$, so~$\ell / k \leq 1/4 \leq 2/7$.

    \item
    $\mathrm{height}(G[\bigtriangleup \cup V^-]) = 4$: \\
    Partition~$V^-$ into~$V^- = \{v\} \cup V^-_1 \cup V^-_2 \cup V^-_3$, where~$V^-_1,V^-_2,V^-_3$ are the vertices in the interior of $(v,x,y)$, $(v,y,z)$ and $(v,z,x)$, respectively.
    At least one of them has cardinality three, because $\mathrm{height}(G[\bigtriangleup \cup V^-]) = 4$.
    We assume without loss of generality that~$\lvert V^-_2 \rvert = 3$.

    Assume first that~$\lvert V^- \rvert \geq 7$.
    Remove~$V^-_2$ from~$G$ to get a graph~$\widetilde G$.
    Then~$\widetilde G$ fulfills the condition of either Case~\ref{itm:stacked4VertexSubtree}, Case~\ref{itm:stackedMaxHeight3} or Case~\ref{itm:stackedHeight4_2vertex} and can be treated as described there.
    In the corresponding case another~$\widetilde k \geq 4$ vertices are removed from~$\widetilde G$ to get~$G'$ and~$\widetilde\ell = 1$ extra edge is needed to extend an edge guard set~$\Gamma'$ for~$G'$ to an edge guard set~$\widetilde \Gamma$ for~$\widetilde G$.
    After reinserting the vertices in~$V^-_2$ we need only one extra edge to extend~$\widetilde\Gamma$ to an edge guard set~$\Gamma$ for~$G$ by Lemma~\ref{lem:stacked6}, because~$G[\{v,y,z\} \cup V^-_2]$ is a~$6$-vertex stacked triangulation.
    In total we get~$k = \widetilde k + 3 \geq 7$ and~$\ell = \widetilde \ell + 1 = 2$, so~$\ell / k \leq 2 / 7$.

    Now assume that~$\lvert V^- \rvert \leq 6$.
    Since~$\lvert V^-_2 \rvert = 3$ it must be~$\lvert V^-_1 \rvert = \lvert V^-_3 \rvert = 1$.
    Remove all vertices in~$V^-$ and add two new vertices using Lemma~\ref{lem:stackedWeakForcingLemma} to get a graph~$G'$, such that there is an edge guard set~$\Gamma'$ for~$G'$ with~$y,z \in V(\Gamma')$.
    By Lemma~\ref{lem:stacked6} there is another edge~$e$ containing~$v$ as an endpoint such that all faces of~$G[\{v,y,z\} \cup V^-_2]$ are guarded.
    Then~$\Gamma := \Gamma' \cup \{e\}$ also doubly guards triangles~$(v,x,y)$ and~$(v,z,x)$, so all faces inside triangle~$(x,y,z)$ are guarded.
    In this case we get~$k = 6 - 2 = 4$ and~$\ell = 1$, so~$\ell / k = 1/4 \leq 2/7$.
    \qedhere
  \end{enumerate}
\end{proof}

\noindent
The inductive proof of Theorem~\ref{thm:stackedUpperBound} can be transformed into an efficient algorithm:

\begin{corollary}
  An edge guard set of size~$\lfloor 2n/7 \rfloor$ for an~$n$-vertex stacked triangulation can be computed in time~$O(n)$.
\end{corollary}

\noindent
We observe that an~$n$-vertex stacked triangulation contains~$3(n - 3) + 1$ triangles whose nesting can be described by a rooted ternary tree~$T$.
Within linear time we can compute how many vertices each of them contains in its interior in~$G$.
Doing this to find~$\bigtriangleup$ in every subgraph during the induction leads to a naive algorithm with quadratic time.

To achieve linear time, we compute this information only once for~$G$.
Now remember that only triangles with at most ten (and at least four) inner vertices are possible candidates for~$\bigtriangleup$.
Add all those that do not contain any other candidate in their interior into a priority queue~$PQ$.
Here~$PQ$ can be implemented to have constant time \emph{insert}- and \emph{deleteMin}-operations, as only seven distinct values are possible.
For each subgraph~$G'$ of~$G$ we choose the triangle~$\bigtriangleup$ as the minimal element in~$PQ$ and pop it.
When removing the interior vertices of~$\bigtriangleup$, we can update the count of interior vertices of the triangle~$\bigtriangleup_p$ that is the parent of~$\bigtriangleup$ in~$T$.
If~$\bigtriangleup_p$ now has between four and ten interior vertices an no other such triangle in its interior, add it into~$PQ$.
If~$\bigtriangleup_p$ has at most three remaining interior vertices, recursively consider its parent in~$T$.

\section{Conclusion \& Open Problems}
We proved new bounds on the size of edge guard sets for stacked triangulations and quadrangulations.
Considering quadrangulations was motivated by the fact that previous coloring-based approaches for general plane graphs failed on quadrangular faces.
Our upper bound of~$\lfloor n/3 \rfloor$ as well as work from Biniaz et al.~\cite{Biniaz2019} about quadrangular faces that are far apart from each other suggests that the difficulty is not due to the quadrangular faces themselves.
Instead the currently known methods seem to be not strong enough to capture the complexity introduced by a mix of quadrangular and non-quadrangular faces.
Finding tight bounds remains an open question, as our construction needs only $\lfloor (n-2)/4 \rfloor$ edge guards.
We proved that this is best possible for~$2$-degenerate quadrangulations and verified exhaustively and computer assisted that~$\lfloor n/4 \rfloor$ is an upper bound for all quadrangulations with~$n \leq 23$~(master's thesis of the first author~\cite{Jungeblut2019}).
\begin{openproblem}
  How many edge guards are sometimes necessary and always sufficient for quadrangulations?
\end{openproblem}

\noindent
For stacked triangulations we proved tight bounds of~$\lfloor 2n/7 \rfloor$.
By this we identified a non-trivial subclass of triangulations needing strictly less than $\lfloor n/3 \rfloor$ edge guards.
We hope that this can be used to improve the upper bound for general triangulations, for example by combining it with bounds for~$4$-connected triangulations along the lines of~\cite{Heldt2014,Knauer2019}.
There the authors show their claims for~$4$\nobreakdash-connected triangulations before decomposing general triangulations into their~$4$\nobreakdash-connected components.
Those components are stacked inside each other at the separating triangles of the original triangulation.
In~\cite{Jungeblut2019} we present a construction for~$4$-connected triangulations needing~$\lfloor (n-2)/4 \rfloor$ edge guards, but no upper bound is known to us that is better than~$\lfloor n/3 \rfloor$.
\begin{openproblem}
  How many edge guards are sometimes necessary and always sufficient for ($4$\nobreakdash-connected) triangulations?
\end{openproblem}

\noindent
Lastly we want to highlight the open problem for general graphs, namely:
\begin{openproblem}
  Can every~$n$-vertex plane graph be guarded by~$\lfloor n/3 \rfloor$ edge guards?
\end{openproblem}

\section*{Acknowledgements}
We thank Kolja Knauer and Lukas Barth for interesting discussions on the topic.

\bibliographystyle{splncs04}
\bibliography{literature}

\end{document}